\documentclass[
final
]{dmtcs-episciences}

\usepackage{hyperref}
\usepackage{xcolor}

\usepackage[T1]{fontenc}
\usepackage{amsmath}
\usepackage{amsfonts}
\usepackage{amssymb}
\usepackage[shortlabels,inline]{enumitem}
\usepackage{xspace}
\usepackage{pdfsync}
\usepackage{comment}

\newlist{compactenum}{enumerate}{3}
\setlist[compactenum]{topsep=0pt,partopsep=0pt,itemsep=0pt,parsep=0pt}

\usepackage[round]{natbib}

\usepackage[textwidth=3.5cm]{todonotes}
\newcommand{\mynote}[3][]{\todo[caption={\sf #3}, color={%
    \ifnum#2=0 green!20
    \else\ifnum#2=1 orange!30
    \else\ifnum#2=2 yellow!20
    \else\ifnum#2=3 cyan!20
    \else magenta!20\fi\fi\fi\fi}, size=\tiny, #1]{\renewcommand{\baselinestretch}{1}\selectfont\sf#3}\xspace}

\numberwithin{equation}{section}

\makeatletter
\def\@proofn[#1]{\noindent\textbf{Proof #1:} \ignorespaces}
\def\@@proofn{\noindent\textbf{Proof:} \ignorespaces}
\AtBeginDocument{
  \def\proofn{\medskip\@ifnextchar[\@proof\@@proof}
  
}
\makeatother
\newcommand{\qedhere}{\qed}
\newcommand{\segment}[1]{\paragraph{#1}}

\newcommand\subword{scattered substring\xspace}
\newcommand\infix{substring\xspace}
\newcommand\subwords{scattered substrings\xspace}
\newcommand\infixes{substrings\xspace}

\newcommand{\N}{\mathbb{N}}

\newcommand{\T}{\mathcal{T}}
\newcommand{\Z}{\mathbb{Z}}
\newcommand{\C}{\ensuremath{\mathcal{C}}\xspace}
\newcommand{\D}{\mathcal{D}}
\renewcommand{\S}{\ensuremath{\mathcal{S}}\xspace}
\renewcommand{\L}{\mathcal{L}}

\newcommand{\PTL}{\textsc{PTL}\xspace}
\newcommand{\ptl}{\PTL}

\newcommand{\eps}{\varepsilon}
\newcommand{\trans}[1]{\stackrel{#1}{\longrightarrow}}

\newcommand{\Dclosure}[1]{#1\mathord{\downarrow}}

\newtheorem{theorem}{Theorem}[section]
\newtheorem{lemma}[theorem]{Lemma}
\newtheorem{corollary}[theorem]{Corollary}
\newtheorem{definition}[theorem]{Definition}

\newtheorem{proposition}[theorem]{Proposition}
\newtheorem{example}[theorem]{Example}

\DeclareSymbolFont{bbold}{U}{bbold}{m}{n}
\DeclareSymbolFontAlphabet{\mathbbold}{bbold}

\newcommand\content[1]{\ensuremath{\contentmorphism(#1)}}

\newcommand\arity[1]{\ensuremath{\text{arity}(#1)}}
\newcommand\contentmorphism{\ensuremath{\text{Alph}}}

\newcommand\LangExp[4]{{#1}_0(\exactcontent{#2_1})^{#3}{#1}_1\cdots {#1}_{{#4}-1}(\exactcontent{{#2}_{#4}})^{#3}{#1}_{#4}}
\newcommand\Lang[3]{{\ensuremath{\mathcal{L}}(\overrightarrow{#1},\overrightarrow{#2},#3)}}

\newcommand\exactcontent[1]{\ensuremath{{#1}^{\circledast}}\xspace}

\author{Wojciech Czerwi\'nski\affiliationmark{1}\thanks{Supported by Poland's National Science Centre grant no.\ UMO-2013/11/D/ST6/03075.}
  \and Wim Martens\affiliationmark{2}\thanks{Supported by DFG grant MA 4938/2-1.}
  \and Lorijn van Rooijen\affiliationmark{3}\thanks{Supported by Agence Nationale de la Recherche ANR 2010 BLAN 0202 01 FREC.}\\
  \and Marc Zeitoun\affiliationmark{4}\thanks{Supported by Agence Nationale de la Recherche, ANR 2010 BLAN 0202 01 FREC and ANR-16-CE40-0007 DeLTA.}
  \and Georg Zetzsche\affiliationmark{5}\thanks{Supported by a fellowship within the PostDoc-Program of the German Academic Exchange Service (DAAD) and by
Labex DigiCosme, Univ.\ Paris-Saclay, project VERICONISS.}
}

\title{A Characterization for Decidable Separability by Piecewise Testable
  Languages\thanks{This article is an extended version of \citep{CzerwinskiMRZ-fct15}.}}

\affiliation{
  University of Warsaw\\
  University of Bayreuth\\
  Paderborn University\\
  LaBRI, University of Bordeaux\\
  LSV, CNRS \& ENS Paris-Saclay
}

\received{2015-11-20}
\accepted{2017-9-25}

\begin{document}

\publicationdetails{19}{2017}{4}{1}{1335}

\maketitle

\renewcommand{\thefootnote}{\arabic{footnote}}

\begin{abstract}
  The separability problem for word languages of a class $\C$ by languages of
  a class $\S$ asks, for two given languages~$I$ and~$E$ from $\C$, whether
  there exists a language $S$ from $\S$ that includes $I$ and excludes $E$,
  that is, $I \subseteq S$ and $S\cap E = \emptyset$. In this work, we assume
  some mild closure properties for $\C$ and study for which such classes 
  separability by a piecewise testable language (\ptl) is decidable. We
  characterize these classes in terms of decidability of (two variants of) an
  unboundedness problem. From this, we deduce that separability by \ptl is
  decidable for a number of language classes, such as the context-free
  languages and languages of labeled vector addition systems. Furthermore, it
  follows that separability by \ptl is decidable if and only if one can compute
  for any language of the class its downward closure wrt.\ the \subword
  ordering (\emph{i.e.}, if the set of \subwords of any language of the class
  is effectively regular).
  
  The obtained decidability results contrast some undecidability
  results. In fact, for all (non-regular) language classes that we
  present as examples with decidable separability, it is undecidable
  whether a given language is a \ptl itself.
  
  Our characterization involves a result of independent interest,
  which states that for \emph{any} kind of languages $I$ and $E$,
  non-separability by \ptl is equivalent to the existence of common patterns
  in $I$ and $E$.
\end{abstract}

\section{Introduction}\label{sec:intro}

Given three languages $I,E,S$, we say that $I$ is \emph{separated} from $E$ by
$S$ if $S$ includes $I$ and excludes $E$, that is, $I
\subseteq S$ and $S \cap E = \emptyset$. In this case, we say that $S$
is a \emph{separator}.  We study the separability problem
of a class~$\C$ by a class $\S$:

  \begin{center}
  \begin{tabular}{ll}
    Given: & Two languages $I$ and $E$ from a class $\C$. \\
    Question: & Can $I$ and $E$ be separated by some language from $\S$?
  \end{tabular}
 \end{center}
 When \C is understood, we also say \S-separability for separability of \C by \S.

 Separability is a classical problem in mathematics and computer science that recently
 found much new interest. For example, recent work investigated the separability
 problem of regular languages by piecewise testable languages~\citep{DBLP:conf/mfcs/PlaceRZ13, DBLP:conf/icalp/CzerwinskiMM13}, by several
 other levels in the quantifier alternation hierarchy of first-order
 logic~\citep{DBLP:conf/icalp/PlaceZ14,pseps3,pzdd2}, by locally testable and locally
 threshold testable languages~\citep{DBLP:conf/fsttcs/PlaceRZ13,pvzltt} and by first order
 definable languages~\citep{pzfo,pzfoj}. Another remarkable example goes beyond regular
 languages: the proof of \citet{DBLP:journals/corr/abs-1009-1076} for the decidability
 of reachability for vector addition systems or Petri nets greatly simplifies earlier
 proofs by \citet{Mayr-stoc81,Mayr-sicomp84} and \citet{Kosaraju-stoc82}, by
 establishing a crucial separation result: non-reachability of a marking ``\textit{goal}''
 from a marking ``\textit{start}'' can be witnessed by a semilinear set containing
 all markings reachable from \textit{start} and excluding \textit{goal}.

 In this paper, we focus on the theoretical underpinnings of separation by piecewise testable languages. Our interest in piecewise testable languages is mainly because of
 the following two reasons.  First, piecewise testable languages form a natural
 class in the sense that they only reason about the \emph{order} of symbols. More
 precisely, they are finite Boolean combinations of regular languages of the form
 $A^* a_1 A^* a_2 A^* \cdots A^* a_n A^*$ in which $a_i \in A$ for every
 $i = 1,\ldots,n$. This definition generalizes to tree languages, for which
 \ptl-separability has already been solved by~\citet{sgl-ptlregtree}.  We are
 investigating to which extent piecewise testable languages and fragments thereof can
 be used for computing simple explanations for the behavior of complex systems
 \citep{HofmanM-icdt15}.

 Second, it was shown recently \citep{DBLP:conf/mfcs/PlaceRZ13,
   DBLP:conf/icalp/CzerwinskiMM13} that separability of regular word languages (given
 by non-deterministic automata) by piecewise testable languages is in PTIME, a
 situation which is quite uncommon. The surprising tractability of this problem is another motivation
 for further investigating the class of piecewise testable languages.

\segment{Separation and Characterization.}  For classes $\C$ effectively closed
under complement, separation of $\C$ by $\S$ is a generalization of the
\emph{characterization problem} of $\C$ by $\S$, often also called
\emph{membership problem}, which is defined as follows: for a given language $L$ of $\C$
decide whether $L$ is in $\S$. Indeed, $L$ is in $\S$ if and only if $L$ can be
separated from its complement by a language from $\S$. The characterization
problem is well studied. The starting points were famous works of
\citet{DBLP:journals/iandc/Schutzenberger65a} and
\citet{DBLP:conf/automata/Simon75}, who solved it for the regular
languages by the first-order definable languages and piecewise testable languages, respectively. There were many more results showing that, for regular
languages and subclasses thereof, often corresponding to a logical fragment, the
problem is decidable \citep[see for example][]{BSlocalConf,
  DBLP:journals/jcss/Zalcstein72,
  DBLP:journals/mst/McNaughton74,knast83,DBLP:conf/stacs/Arfi87,
  DBLP:journals/mst/PinW97,DBLP:journals/tcs/Straubing88, twfodeux, gssig2,
  Tesson02diamondsare, KlimaPTL11,DBLP:conf/fossacs/CzerwinskiDLM13,
  DBLP:conf/icalp/PlaceZ14,pzdd2,pseps3,ABOKK_delta_n}. Similar problems have been
considered for
trees~\citep{DBLP:conf/concur/BojanczykI09,Benedikt:2009:RTL:1614431.1614435,
  DBLP:journals/corr/abs-1208-5129,AntonopoulosHMN-icdt12,psfo2}.

Obtaining a decidable characterization for a class is considered as a way to get a fine
understanding of the class. In particular, solving a characterization problem requires
a proof that a language satisfying some decidable property belongs to the class, which
usually yields a canonical construction for the language. For instance, in the case of
a logical class, this gives a canonical sentence defining any input language that
fulfills the condition we want to prove as a decidable characterization. Recently,
several generalizations of the characterization problem have been introduced as a means
to obtain even more information about the class $\S$ under study. Such generalizations
amount to \emph{approximating} languages of $\C$ by languages of~$\S$
\citep{pz-covers:mfcs16}. In turn, such information may be exploited when studying more
complex classes, \emph{e.g.}, classes that are higher in the quantifier alternation
hierarchy~\citep{DBLP:conf/icalp/PlaceZ14}.
Separation is the simplest of such approximation problems: it asks to
over-approximate the first language,~$I$, by a language $S$ in $\C$, while the second
language, $E$, serves as a quality measure of  this approximation.

\segment{Beyond regular languages.}
To the best of our knowledge, all the above work and in general all the
decidable characterizations were obtained in cases where $\C$ is the class of
regular languages, or a subclass of it.  This could be due to several negative
results which may seem to form a barrier for any nontrivial decidability beyond
regular languages.  In this work, we consider language classes beyond the
regular languages, but we assume them to be full trios (meaning they satisfy
some mild closure properties). Beyond the regular languages, we quickly
encounter undecidability of the problems above.  For instance, for a
context-free language (given by a grammar or a pushdown automaton) it is
well known that it is undecidable to determine whether it is a regular
language, by \citeauthor{Greibach-mst68}'s theorem (\citeyear{Greibach-mst68}).
In the same way, one can show that for every full trio that contains the
language $\{a^n b^n \mid n\ge 0\}$, it is undecidable to determine whether a
given language of the full trio is piecewise testable (see Section~\ref{sec:undecidable-classes}).

In the case of context-free languages, there is a strong connection
between the intersection emptiness problem and
separability. Trivially, testing intersection emptiness of two given
context-free languages is the same as deciding if they can be
separated by some context-free language. However, in general, the
negative result is even more overwhelming.
\citet{SzymanskiW-sicomp76} proved that separability of
context-free languages by \emph{regular} languages is undecidable. This was
then generalized by \citet{DBLP:journals/jacm/Hunt82a}, who proved
that separability of context-free languages by \emph{any} class containing
all the \emph{definite languages} is undecidable. A language $L$ is
\emph{definite} if it can be written as $L = F_1 A^* \cup F_2$, where
$F_1$ and $F_2$ are finite languages over alphabet $A$. As such, for
definite languages, it can be decided whether a given word $w$ belongs
to $L$ by looking at the prefix of $w$ of a given fixed length. The
same statement holds for \emph{reverse definite} languages, in which
we are looking at suffixes.  Containing all the definite, or reverse
definite, languages is a very weak condition.  Note that if a logic
can test what is the $i$-th symbol of a word and is closed under Boolean
combinations, it can already define all the definite languages.  In his
paper, Hunt~III makes an explicit link between intersection emptiness and
separability. Hunt~III writes: \emph{%
  ``We show that separability is undecidable in general for the same
  reason that the emptiness-of-intersection problem is
  undecidable. Therefore, it is unlikely that separability can be used
  to circumvent the undecidability of the emptiness-of-intersection
  problem.''}


\segment{Our Contribution.} In this paper, we show that the above mentioned
quote does not apply for separability by a piecewise testable language
(\ptl): we characterize those full trios for which separability by a \ptl is
decidable. A \emph{full trio} is a nonempty language class that is closed
under rational transductions, or, equivalently by Nivat's
theorem~\citep[see][]{Nivat_1968,Berstel:Transductions-context-free-languages:1979:a},
closed under direct and inverse homomorphic images of free monoids,
and under intersections with regular languages.

Our characterization states that separability by \ptl is decidable if and only
if (one of two variants of) an unboundedness problem is decidable. This yields
decidability for a range of language classes, such as those with effectively
semilinear Parikh images and the languages of labeled vector addition systems.
In particular, this means that separability by \ptl is decidable for
context-free languages, and, according to very recent
results, also for the languages of higher-order
pushdown automata~\citep{HagueKochemsOng2016} and even higher-order recursion schemes~\citep{ClementeParysSalvatiWalukiewicz2016}.

The two variants of the unboundedness problem are called
\emph{diagonal problem} and \emph{simultaneous unboundedness problem
(SUP)}, of which the latter is ostensibly easier than the former.  Our
reduction of separability by \ptl consists of three steps:
\begin{itemize}
\item[--] In the first step, we show that (arbitrary) languages $I$ and $E$
  are \emph{not} separable by \ptl if and only if they possess a certain \emph{common
    pattern}.
\item[--]  In the second step we use this fact to reduce the \ptl-separation
  problem to the diagonal problem.
\item [--] In the last step, we employ ideal decompositions for downward closed languages to
  reduce the diagonal problem to the SUP.
\end{itemize}

\smallskip
For the converse direction, we directly reduce the SUP to separability by \ptl.
\medskip

Another consequence of our characterization is a connection to the
problem of \emph{computing downward closures}. It is well known that
for any language $L$, its \emph{downward closure}, \emph{i.e.},~the set of
\subwords of members of $L$, is a regular
language~\citep{Haines1969}. This is a straightforward consequence of Higman's
Lemma~\citep{Higman_1952}. However, given a language $L$, it is not
always possible to compute a finite automaton for the downward closure
of $L$. Since the downward closure appears to be a useful abstraction,
this raises the question of when we can compute it. Our
characterization here implies that for each full trio $\S$, downward
closures of languages of $\S$ are computable if and only if separability of
languages of $\S$ by \ptl is decidable.

A curiosity of this work is perhaps the absence of algebraic methods, as most
decidability results we are aware of have considered syntactic monoids of regular
languages and investigated properties thereof. In the setting of this paper, the
situation is different since we want to separate input languages that are not regular
(hence, whose syntactic monoid is infinite), such as context-free languages. This would
make it difficult to design any algebraic framework for them. However, the main reason
why we do not rely on syntactic monoids is that the characterization and separation
problems are of very different nature. For characterization, whether the input language
belongs to $\C$ indeed boils down to checking a property of the syntactic (ordered)
monoid for the classes $\C$ that have been investigated so far, namely, (positive)
varieties of regular languages. For separation however, the language we are looking for
in the class $\C$ is \emph{not} one of the inputs, so that a given language from $\C$
can serve as a separator for several unrelated instances of the $\C$-separation
problem, and these instances may well have completely different syntactic properties.

%

%



\segment{Structure of the paper} In Section~\ref{sec:prelim} we introduce basic notions
and notation, the diagonal and SUP problems, and state our main result, which connects
these problems with separation by \ptl. Section~\ref{sec:patterns} gives a simple
criterion by ``common patterns'' for two arbitrary languages to be separable by~\ptl.
Section~\ref{sec:separability} is devoted to proving equivalent conditions for
separability by \ptl: the existence of common patterns, the diagonal problem, the SUP
problem, and the computation of downward closures. In Section~\ref{sec:feasibility} we
discuss applications of our result, in particular we show that it applies to
context-free languages and languages of labeled vector addition systems (alternatively,
languages of labeled Petri nets). Finally, in Section~\ref{sec:undecidable-classes}, we
present some classes for which all these problems are undecidable.



\section{Preliminaries and Main Results}\label{sec:prelim}
We assume the reader to be familiar with regular expressions and regular
languages. In this section, we first set the notation and state our main
results.

The set of all integers and nonnegative integers are denoted by $\Z$ and $\N$
respectively. An \emph{alphabet} is a finite nonempty set, which we usually
denote with $A$, $B$, $C$, \ldots\ or indexed versions thereof. We refer to elements of $A$
as \emph{symbols}.

A \emph{word} is a (possibly empty) concatenation $w = a_1 \cdots a_n$ of
symbols $a_i$ that come from an alphabet~$A$. A \emph{\infix} of $w$ is a
sequence $a_j a_{j+1} \cdots a_{j+k}$ of consecutive symbols of $w$.  The
\emph{length} of a word $w=a_1\cdots a_n$ is $n$, that is, the number of its
symbols. The \emph{alphabet of a word $w=a_1\cdots a_n$} is the set
$\{a_1,\ldots,a_n\}$ and is denoted $\content{w}$. The empty word, of length 0
and of alphabet $\emptyset$, is denoted by $\varepsilon$.

For a subalphabet $B \subseteq A$, a word $v \in A^*$ is a \emph{$B$-\subword} of $w$, denoted
$v \preceq_B w$, if $v$ can be obtained from $w$ by removing symbols from $B$, that is, if
$v = b_1 \cdots b_m$ and $w \in B^* b_1 B^* \cdots B^* b_m B^*$. Notice that we do not require that
$\{b_1,\ldots,b_m\} \subseteq B$ or $B \subseteq \{b_1,\ldots,b_m\}$. We simply refer to
$A$-\subwords as \emph{\subwords} and refer to the relation $\preceq_A$ as the \emph{\subword}
relation, denoted by $\preceq$. A regular word language over alphabet $A$ is a \emph{piece language}
if it is of the form $A^* a_1 A^* \cdots A^* a_n A^*$ for some $a_1, \ldots, a_n \in A$, that is, it is the
set of words having $a_1 \cdots a_n$ as a \subword.  A regular language is a \emph{piecewise testable}
language (\ptl) if it is a finite Boolean combination of piece languages. We denote the class
of all piecewise testable languages also by \ptl.

\subsection{Separability and Common Patterns}

The first main result of the paper proves that 
two (not necessarily regular) languages are not separable by \ptl if and only if they have a common
pattern. We now make this more precise.

A \emph{factorization pattern} is an element of
$(A^*)^{p+1}\times(2^A\setminus\emptyset)^p$ for some $p\geq0$. In other terms, if
$(\overrightarrow{u}, \overrightarrow{B})$ is such a factorization pattern,
there exist words $u_0,\ldots,u_p\in A^*$ and nonempty alphabets
$B_1,\ldots,B_p\subseteq A$ such that $\overrightarrow{u}=(u_0,\ldots,u_p)$ and
$\overrightarrow{B}=(B_1,\ldots,B_p)$.
For $B\subseteq
A$, we denote by $\exactcontent B$ the set of words with alphabet
exactly $B$, that is,
\[
  \exactcontent{B} = \{ w \in B^* \mid \content{w} = B\}.
\]
Given a factorization pattern $(\overrightarrow{u},\overrightarrow{B})$ with
$\overrightarrow{u}=(u_0,\ldots,u_p)$ and $\overrightarrow{B}=(B_1,\ldots,B_p)$, define
\begin{equation*}
  {\Lang u B n =\LangExp  u B n p.}
\end{equation*}
\noindent In other terms, in a word of $\Lang u B n$, the infix between $u_{k-1}$ and $u_k$
is required to be the concatenation of $n$ words over $B_k$, each containing \emph{all}
symbols of $B_k$ (for each $1\leq k\leq p$).
An infinite sequence $(w_n)_n$ is said to be \emph{$(\overrightarrow{u},\overrightarrow{B})$-adequate} if
\begin{equation*}
  \forall n \in \N,\ w_n\in\Lang u B n.
\end{equation*}

\begin{example}
  \emph{As an example, consider the factorization pattern $(\overrightarrow{u},\overrightarrow{B})$ with $\overrightarrow{u}=(\varepsilon,c,\varepsilon)$ and $\overrightarrow{B}=(\{a,b\}, \{a\})$. Then, $\Lang u B n = (\exactcontent{\{a,b\}})^n c (\exactcontent{\{a\}})^n$.
    A sequence of words $w_1, w_2, \ldots$ is $(\overrightarrow{u},\overrightarrow{B})$-adequate if first of all
    for every $n \in \N$, there is a single symbol $c$ in word $w_n$.
    Second of all, after this symbol $c$ there are only symbols $a$ and, moreover, at least $n$ of them.
    Finally, the prefix of $w_n$ before the distinguished symbol $c$ contains only symbols $a$ and $b$ and can be split
    into at least $n$ words such that all of them contain at least one symbol $a$ and at least one symbol $b$.}
\end{example}
Finally, we say that language $L$ \emph{contains the pattern $(\overrightarrow{u},
  \overrightarrow{B})$} if there exists an infinite sequence of words $(w_n)_n$ in $L$ that is $(\overrightarrow{u},\overrightarrow{B})$-adequate.
We can now formally state our first main theorem which we will prove in Section~\ref{sec:patterns}:
\begin{theorem}  \label{lem:common-ub-path-2} \label{theo:patterns}
  Two word languages $I$ and $E$ are not separable by \ptl if and only
  if they contain a common pattern $(\overrightarrow{u},
  \overrightarrow{B})$.
\end{theorem}

\subsection{Characterizations for Decidable Separability}
The second main result is a set of characterizations that say, for \emph{full
  trios}, when
separability by piecewise testable languages is decidable. Full
trios, also called \emph{cones}, are language classes that are closed
under three operations that we recall next~\citep{Berstel:Transductions-context-free-languages:1979:a,GinsburgG-swat67}.

Fix a language $L$ over alphabet $A$.
For an alphabet $B$, the \emph{$B$-projection} of a word
is its longest \subword
consisting of symbols from $B$. The \emph{$B$-projection} of a language $L$ is the
set of all $B$-projections of words belonging to $L$. Therefore, the
$B$-projection of $L$ is a language over alphabet $A \cap B$. The
\emph{$B$-upward closure} of a language $L$ is the set
of all words that have a $B$-\subword in $L$, \emph{i.e.},
\begin{equation*}
  \bigl\{w \in (A \cup B)^* \mid \exists v \in L \text{ such that } v \preceq_B w\bigr\}.
\end{equation*}

\noindent In other words, the $B$-upward closure of $L$ consists of
all words that can be obtained by taking a word in~$L$ and padding it
with symbols from $B$.

A \emph{language class} is a collection of languages that contains at
least one nonempty language.  A language class $\C$ is \emph{closed}
under an operation $\textsc{op}$ if $L \in \C$ implies that
$\textsc{op}(L) \in \C$.  We use the term \emph{effectively closed}
if, furthermore, the representation of $\textsc{op}(L)$ can be
effectively computed from the representation of $L$. A class $\C$ of
languages is a \emph{full trio} if it is effectively
closed under:
\begin{compactenum}[1.]
\item $B$-projection for every finite alphabet $B$,
\item $B$-upward closure for every finite alphabet $B$, and
\item intersection with regular languages.
\end{compactenum}

\smallskip We note that full trios are usually defined differently, either through
closures under direct and inverse images of morphisms and intersection with
regular languages, or through closure under rational transductions
\citep{GinsburgG-swat67,Berstel:Transductions-context-free-languages:1979:a}.
However, we use the above mentioned properties in the proofs, which are easily
seen to be equivalent.

We will now introduce three problems whose decidability or computability will
be equivalent to decidability of separability by piecewise testable languages:
the \emph{diagonal problem}, the \emph{simultaneous unboundedness problem}, and
the \emph{downward closure problem}.

\segment{Downward Closure Problem.}

For a language $L\subseteq A^*$, its \emph{downward closure $\Dclosure{L}$} is
defined as the set of all \subwords of words in $L$, that is,
$$\Dclosure{L} = \{ u\in A^* \mid \exists v\in L\colon u\preceq
v\}.$$ We say that $L$ is downward closed if $L=\Dclosure{L}$.
It is well-known that the downward closure of any language $L$ is regular
\citep{Haines1969}. This is a direct consequence of Higman's
Lemma~\citep{Higman_1952}, which states that $\preceq$ is a well quasi
ordering\footnote{Actually, the relation $\preceq$ that we defined over $A^{*}$
  is even a well \emph{ordering}.} (WQO) on~$A^*$, \emph{i.e.}, that any infinite
sequence of words admits an infinite $\preceq$-increasing subsequence. Indeed,
the complement of a downward closure is closed under $A$-upward closure. It
follows from Higman's Lemma that this complement has a finite number of $\preceq$-minimal
elements and is therefore a finite union of piece languages.

\smallskip
Downward closed languages play an important role in the theory of lossy
channel systems~\citep{Finkel_Wsts87,LossyAJ96,wsts01}, which feature
communication channels where messages can be dropped arbitrarily. These have
been investigated extensively, because the WQO property of the subword
relation yields decidability results that fail in the non-lossy case. However,
the complexity of such decidable problems may be huge
\citep{lossy_nonprim,ShSchpower13}.

Moreover, the downward closure is a useful abstraction of languages: Suppose a
language $L$ describes the set of action sequences of a system modeled, for
example, by a vector addition system or a pushdown automaton. If this system is
observed via a lossy channel, then $\Dclosure{L}$ is the set of sequences seen
by the observer~\citep{HabermehlMeyerWimmel2010}. Furthermore, computing a
regular representation of $\Dclosure{L}$ for a given language $L$ is in many
situations sufficient for safety verification of parametrized asynchronous
shared-memory systems~\citep{LaTorreMuschollWalukiewicz2015}.

Starting from a regular language given, \emph{e.g.},~by a finite automaton, one can
easily compute a representation of its downward closure (one can even obtain
precise upper and lower bounds in terms of the size of an automaton recognizing
it~\citep{Karandikar_2015}). 
However, despite the fact that the
downward closure of a language is always a simple regular language (the
complement of a union of piece languages), it is not always possible to
effectively compute a finite automaton for $\Dclosure{L}$ given a description
of a (possibly nonregular) language $L$. Such negative results are known, for example, for Church-Rosser
languages~\citep{GruberHolzerKutrib2007} and reachability sets of lossy channel
systems~\citep{Mayr2003}. We say that \emph{downward closures are computable
for $\C$} if a finite automaton for $\Dclosure{L}$ can be algorithmically
computed for every language $L$ in $\C$.

\segment{Unboundedness Problems.}
Let $A = \{a_1, \ldots, a_n\}$ be ordered with $a_1<\cdots<a_n$. For a symbol $a \in A$ and a word $w
\in A^*$, let $\#_a(w)$ denote the number of occurrences of $a$ in
$w$.  The \emph{Parikh image} of a word $w$ is the $n$-tuple
$(\#_{a_1}(w), \ldots, \#_{a_n}(w))$.  The \emph{Parikh image} of a
language $L$ is the set of all Parikh images of words from $L$. A
tuple $(m_1,\ldots,m_n) \in \N^n$ is \emph{dominated} by a tuple
$(d_1,\ldots,d_n) \in \N^n$ if $d_i \geq m_i$ for every $i =
1,\ldots,n$.
\begin{definition}
  The \emph{diagonal problem for $\C$} is the following decision problem:
  \begin{description}
  \item[Input] A language $L\subseteq A^*$ from $\C$.
  \item[Question] Is each tuple $(m, \ldots, m)\in\N^n$ dominated by some tuple in the Parikh image of $L$?
  \end{description}
\end{definition}
Equivalently, the diagonal problem for $\C$ asks whether there are \emph{infinitely
many} tuples  $(m, \ldots, m)$ dominated by some tuple in the Parikh image of $L$.

The \emph{simultaneous unboundedness problem (SUP)} is a restricted version of
the diagonal problem where the input is a language in which the symbols are
grouped together:
\begin{definition} The \emph{simultaneous unboundedness problem (SUP) for $\C$} is the
  following decision problem:
  \begin{description}
  \item[Input] A language $L\subseteq a_1^*\cdots a_n^*$ \,from $\C$, for
    some ordering $a_1,\ldots,a_n$ of $A$.
  \item[Question]  Is each tuple $(m, \ldots, m)\in\N^n$ dominated by some tuple in the Parikh image of $L$?
  \end{description}
\end{definition}
In other words, the SUP asks
whether
$\Dclosure{L}=a_1^*\cdots a_n^*$.


\segment{Characterizations.}
Perhaps surprisingly, our second main result here implies that in a
full trio, separability by \ptl is decidable if and only if downward
closures are computable.  Our proof relies on the following
characterization of computability of downward closures:
\begin{theorem}[\citeauthor{Zetzsche-icalp15}, \citeyear{Zetzsche-icalp15}]\label{theo:dc}
  Let $\C$ be a full trio. Then downward closures are computable for
  $\C$ if and only if the SUP is decidable for $\C$.
\end{theorem}
We are now ready to state the second main result, which is an
extension of Theorem~\ref{theo:dc} and connects it to separability:
\begin{theorem}\label{theo:decidability}
  For each full trio $\C$, the following are equivalent:
  \begin{enumerate}[label=(\arabic*), ref=(\arabic*)]
  \item\label{theo:it:sep} Separability of $\C$ by \ptl is decidable.
  \item\label{theo:it:diag} The diagonal problem for $\C$ is decidable.
  \item\label{theo:it:sup} The SUP for $\C$ is decidable.
  \item\label{theo:it:dc} Downward closures are computable for $\C$.
  \end{enumerate}
\end{theorem}
Theorem~\ref{theo:dc} provides the equivalence between \ref{theo:it:sup} and
\ref{theo:it:dc}. We prove the remaining equivalences in
Section~\ref{sec:separability}. In particular, we present an algorithm to
decide separability for full trios that have a decidable diagonal problem,
showing one direction of the equivalence.  The algorithm does not rely on
semilinearity of Parikh images.  For example, in Section~\ref{sec:feasibility}
we apply the theorem to Vector Addition System languages, which do not have a
semilinear Parikh image. 


\section{Common Patterns}\label{sec:patterns}

In this section we prove Theorem~\ref{theo:patterns}.  We say that an infinite
sequence is \emph{adequate} if it is $(\overrightarrow{u},
\overrightarrow{B})$-adequate for some factorization pattern $(\overrightarrow{u},
\overrightarrow{B})$. We will show the
following combinatorial statement using Simon's Factorization Forest
Theorem~\citep{Simon-tcs90}. 

\begin{lemma}\label{lem:extract-adequate-subseq}
  Every infinite sequence $(w_n)_n$ of words admits an adequate subsequence.
\end{lemma}

Before proving it, we note that Lemma~\ref{lem:extract-adequate-subseq} gives us
an alternative proof of Higman's Lemma.

\begin{corollary}[\citeauthor{Higman_1952}, \citeyear{Higman_1952}]
  The scattered subword ordering over $A^{*}$ is well.
\end{corollary}

\begin{proofn}
  We want to show that for any sequence of words $(w_{n})_{n\in\mathbb{N}}$, there exist
  two indices $i,j$ such that $i<j$ and $w_{i}\preceq w_{j}$, where $\preceq$
  denotes the scattered substring relation. Taking a subsequence of
  $(w_{n})_{n\in\mathbb{N}}$ if necessary, we may assume by
  Lemma~\ref{lem:extract-adequate-subseq} that $(w_{n})_{n\in\mathbb{N}}$ is
  adequate. This means that there exists a factorization pattern
  $(\overrightarrow{u},\overrightarrow{B})$, where
  $\overrightarrow{u}=(u_{0},\ldots,u_{p})$ and
  $\overrightarrow{B}=(B_{1},\ldots,B_{p})$ such that 
  $(w_{n})_{n\in\mathbb{N}}$ is
  $(\overrightarrow{u},\overrightarrow{B})$-adequate, that~is,
  \begin{equation}
    \label{eq:h}
    \forall n \in \N,\quad w_n\in\LangExp  u B n p.
  \end{equation}
  In particular for $n=0$, there exist $v_{0},\ldots,v_{p}$ such that
  \begin{equation}
    \label{eq:w_0}
    w_{0}=u_{0}v_{0}u_{1}\ldots u_{p-1}v_{p}u_{p}
  \end{equation}
  and
  \begin{equation}
    \label{eq:vi}
    v_{i}\in \exactcontent{B_{i}}.
  \end{equation}
  Note that by definition of the $\exactcontent{}$ operation, if
  $v\in\exactcontent{B}$, then for all $v'\in(\exactcontent{B})^{|v|}$, we have
  $v\preceq v'$. Together with~\eqref{eq:h}, \eqref{eq:w_0} and \eqref{eq:vi}, this entails that
  $w_{0}\preceq w_{n}$ for all $n\geq n_0=\max\{|v_{i}|\mid 0\leq i\leq p\}$. It
  then suffices to choose $i=0$ and $j=n_{0}$.\qed
\end{proofn}

\begin{proofn}[of Lemma~\ref{lem:extract-adequate-subseq}]
  We use Simon's Factorization Forest Theorem, which we recall. See
  \citep{Simon-tcs90,Kufleitner-mfcs08,bfacto,Colcombet-tcs10,tc-handbook16}
  for proofs and extensions of this theorem. A
  \emph{factorization tree} of a nonempty word $x$ is a finite ordered unranked tree $T(x)$
  whose nodes are labeled by nonempty words, and such that:
  \begin{itemize}
  \item all leaves of $T(x)$ are labeled by symbols,
  \item all internal nodes of $T(x)$ have at least 2 children, and
  \item if a node labeled $y$ has $k$ children labeled $y_1,\ldots, y_k$
    from left to right, then $y=y_1\cdots y_k$.
  \end{itemize}
  Given a semigroup morphism $\varphi:A^+\to S$ into a finite semigroup
  $S$, such a factorization tree is called \emph{$\varphi$-Ramseyan} if every
  internal node has either 2 children, or $k$ children labeled
  $y_1,\ldots, y_k$, in which case $\varphi$ maps all words $y_1,\ldots,y_k$ to the
  same idempotent of~$S$, \emph{i.e.}, to an element $e$ such that $ee=e$. Simon's
  Factorization Forest Theorem states that 
  every word has a $\varphi$-Ramseyan factorization tree of height at most~$3|S|$.

  Let $(w_n)_n$ be an infinite sequence of words. We use Simon's Factorization
  Forest Theorem with the morphism $\contentmorphism
  :A^+\to2^A$. Recall that $\contentmorphism$ maps a word $w$ to the
  set of symbols used in $w$.
  
  Consider a sequence $(T(w_n))_n$, where $T(w_n)$ is an
  $\contentmorphism$-Ramseyan tree of $w_n$, given by the Factorization Forest
  Theorem. In particular, one may choose $T(w_n)$ of height at most
  $3 \cdot 2^{|A|}$. Therefore, extracting a subsequence if necessary, one may
  assume that the sequence of heights of the trees $T(w_n)$ is a constant~$H$. We
  argue by induction on $H$. If $H=0$, then every $w_n$ is a symbol. Hence, one
  may extract from $(w_n)_n$ a constant subsequence, say $(a)_{n\in\N}$, which is
  therefore $((a), ())$-adequate, that is, adequate for the factorization pattern
  consisting of $(a)$ and the empty tuple $()$. Hence, it is adequate, which
  concludes the case $H=0$.

  Assume now that $H>0$. We denote the arity of the root of $T(w_n)$ by $\arity{w_n}$ and we
  call it the arity of $w_n$. We distinguish two cases:
  \begin{description}
  \item [Case 1] One can extract from $(w_n)_n$ a subsequence of bounded arity. Therefore,
    one may extract a subsequence of constant arity, say $K$, from $w_n$, and replacing
    $(w_n)_n$ by such a subsequence, one may assume that $(w_n)_n$ itself has this
    property. This implies that each $w_n$
    can be written as a concatenation of $K$ words
    $$w_n=w_{n,1}\cdots w_{n,K},$$
    where $w_{n,i}$ is the label of the $i$-th child of the root in
    $T(w_n)$.  Therefore, the $\contentmorphism$-Ramseyan subtree of
    each $w_{n,i}$ is of height at most $H-1$. By induction, one can
    extract from every $(w_{n,i})_n$ an adequate
    subsequence. Proceeding iteratively for $i=1,2,\ldots, K$, one
    extracts from $(w_n)_n$ a subsequence $(w_{\sigma(n)})_n$ such
    that every $(w_{\sigma(n),i})_n$ is adequate.  But a finite
    concatenation of adequate sequences is obviously adequate.
    Therefore, the subsequence $(w_{\sigma(n)})_n$ of $(w_n)_n$ is
    also adequate.
    
  \item [Case 2] The arity of $w_n$ grows to infinity. Therefore, extracting if
    necessary, one can assume for every $n$, that $\arity{w_n}\geq \max(n,3)$.
    Since each factorization tree is $\contentmorphism$-Ramseyan and since the
    arity of the root of each tree is at least 3, all
    children of the root of the $n$-th tree map to the same idempotent in $2^A$, say
    $B_n\subseteq A$. Since $2^A$ is finite, one can further extract a subsequence, say
    $w_{\sigma(n)}$, such that $B_{\sigma(n)}$ is constant, equal to some
    $B\subseteq A$. To sum up, each word of the subsequence is of the form
    $$w_{\sigma(n)}=w_{n,1}\cdots w_{n,K_n},$$
    with $K_n\geq n$ and, where the alphabet of $w_{n,i}$ is
     $B$. Therefore, $w_{\sigma(n)}\in
    (\exactcontent B)^{K_n}\subseteq(\exactcontent B)^n$, which means that
    $(w_{\sigma(n)})_n$ is $((\varepsilon,
    \varepsilon), (B))$-adequate, hence it is adequate.\qedhere
  \end{description}
\end{proofn}


For a nonempty word $w$, denote its first (resp., last) symbol by
$\text{first}(w)$, (resp., $\text{last}(w)$).
A factorization pattern $(\overrightarrow{u},\overrightarrow{B})= ((u_0,\ldots,u_p),(B_1,\ldots,B_p))$ is said to be \emph{proper} if
\begin{enumerate}
\item\label{item:2} 
 for all $i = 0,\ldots,p-1$, we have $u_i=\varepsilon$ or $\text{last}(u_i) \notin B_{i+1}$, 
 \item 
 for all $i = 1,\ldots,p$, we have $u_i=\varepsilon$ or $\text{first}(u_i) \notin B_{i}$,  and
\item\label{item:3}  
  for all $i=1,\ldots,p-1$, if $u_i = \varepsilon$, then we have $\big(B_{i} \nsubseteq B_{i+1} \text{ and } B_{i+1} \nsubseteq B_{i}\big)$.
\end{enumerate}

Note that if a sequence $(w_n)_n$ is adequate, then there exists a \emph{proper} factorization pattern $(\overrightarrow{u}, \overrightarrow{B})$ such that $(w_n)_n$ is $(\overrightarrow{u}, \overrightarrow{B})$-adequate. This is easily seen from the following observations and their symmetric counterparts:
\begin{equation*}
  \begin{array}{rcl}
    u = a_1 \cdots a_k \text{ and } a_k \in B &\ \Rightarrow\ & a_1 \cdots a_k (\exactcontent{B})^n \subseteq a_1 \cdots a_{k-1} (\exactcontent{B})^n, \\
    B_{i} \subseteq B_{i+1} &\ \Rightarrow\ & (\exactcontent{B_{i}})^n (\exactcontent{B_{i+1}})^n \subseteq (\exactcontent{B_{i+1}})^n.
  \end{array}
\end{equation*}

The following lemma gives a condition under which two sequences share
a factorization pattern and is very similar to \cite[Theorem~8.2.6]{JAbook}. In its statement, we write $v
\sim_n w$ for
two words $v$ and $w$ if they have the same \subwords up to length
$n$, that is, if for every word $u$ of length at most $n$, we have $u \preceq v$
if and only if $u \preceq w$. 
%
Notice that $\sim_n$ is an
equivalence relation for every $n \in \N$.

\begin{lemma}
  \label{lem:samepatt}
  Let $(\overrightarrow{u},\overrightarrow{B})$ and $(\overrightarrow{t}, \overrightarrow{C})$ be proper
  factorization patterns.
  Let $(v_n)_{n}$ and  $(w_n)_{n}$  be two sequences of words such that
  \begin{itemize}
  \item $(v_n)_{n}$ is $(\overrightarrow{u},\overrightarrow{B})$-adequate
  \item $(w_n)_{n}$ is $(\overrightarrow{t},
    \overrightarrow{C})$-adequate, and
 \item $v_n \sim_n w_n$  for every $n\geq0$.
 \end{itemize}
 Then, $\overrightarrow{u}=\overrightarrow{t}$ and $\overrightarrow{B} = \overrightarrow{C}$.
\end{lemma}
\begin{proof}
  For a factorization pattern $(\overrightarrow{u},\overrightarrow{B})$, we define
  \[
    \| (\overrightarrow{u},\overrightarrow{B}) \| = (\sum_{i=0}^{\ell} |u_i| )+\ell,
  \] 
  where $\ell+1$ is the number of components in the vector $\overrightarrow{u}=(u_0,\ldots,u_\ell)$. Let 
  \[
    k = \max(\| (\overrightarrow{u},\overrightarrow{B}) \|, \| (\overrightarrow{t},\overrightarrow{C}) \|) +1.
  \] 
  Consider the second word of the sequence $(v_n)_{n}$, \emph{i.e.},~$v_1 = u_0
  r_1 u_1 \cdots r_{\ell} u_\ell$, where $\content{r_i} = B_i$. Define 
  \begin{equation}
    v_1^{(k)} = u^{}_0 r_1^k u^{}_1 \ldots r_{\ell}^k u^{}_\ell.\label{eq:v0}
  \end{equation}
  Recall that $(v_n)_{n}$ being a $(\overrightarrow{u},\overrightarrow{B})$-adequate sequence
  means that
  \begin{equation}
    \forall n\in\N,\quad v_n \in \LangExp  u B n \ell.\label{eq:vn}
  \end{equation}
  Let $N = k \cdot \max(|r_1|, \ldots, |r_n|)$. By \eqref{eq:v0} and \eqref{eq:vn},
  we get $v_1^{(k)} \preceq v_N$. Let $M \geq \max(N, | v_1^{(k)}|)$, so that
  $v_1^{(k)}$ is a \subword of $v_M$ of length at most $M$. Since $v_M \sim_M
  w_M$, this gives that
  \begin{equation}
    \label{eq:embed}
    v_1^{(k)} \preceq w_M.
  \end{equation}
  To show that $(\overrightarrow{u},\overrightarrow{B}) =
  (\overrightarrow{t}, \overrightarrow{C})$, we first  define a bijection
  between elements of the sequence of indexed alphabets in
  $\overrightarrow{B}$ and elements of the sequence of those in
  $\overrightarrow{C}$. For this, we embed
  $v_1^{(k)}$ in $w_{M}$ in the `leftmost' way. Let us make
  this precise: let $x,y\in A^*$ with $x=x_1\cdots x_{|x|}$ and
  $y=y_1\cdots y_{|y|}$, where $x_i,y_i\in A$. If $x\preceq y$, by definition there exists an
  increasing mapping $h:\{1,\ldots,|x|\}\to \{1,\ldots,|y|\}$ such that
  $x_i=y_{h(i)}$ holds for all $i$. We say that such a mapping is an \emph{embedding}
  of $x$ in $y$. The \emph{leftmost embedding} $h_\text{left}$ of $x$ in $y$
  is such that additionally, for any embedding $h$, we have $h_\text{left}(i)\leq
  h(i)$ for all~$i$.

  \smallskip To define the bijection, it is convenient to tag the subalphabets
  $B_{i},C_{j}$, as some of them may be equal. Therefore, we set
  $\mathbf{B} = \{ (B_1, 1),\ldots, (B_{\ell}, \ell)\}$ and
  $\mathbf{C} = \{ (C_1,1), \ldots, (C_{m},m)\}$, where $\ell$ (resp.\ $m$)
  denotes the length of the vector $\overrightarrow{B}$ (resp.\
  $\overrightarrow{C}$). We define a function
  $f : \mathbf{B} \rightarrow \mathbf{C}$ as follows. Consider the leftmost
  embedding $h_\text{left}$ of $v_1^{(k)}$ in $w_{M}$, which exists
  by~\eqref{eq:embed}. Since $(w_n)_n$ is a
  $(\overrightarrow{t},\overrightarrow{C})$-adequate sequence, one may write
  $$w_M=t_0s_1t_1\cdots t_{m-1}s_mt_m\quad \text{with } s_j\in(\exactcontent{C_j})^M
  \text{ for all } j.$$
  For $i$ in $\{1,\ldots,\ell\}$, consider the \infix $r_i^{k}$ of $
  v_1^{(k)}$. 
By definition of $k > \| (\overrightarrow{t},\overrightarrow{C})
  \|$ and since $|r_i| >0$, the
  pigeonhole principle implies that all positions of one of these $k$ copies of
  $r_i$ must be mapped
  by $h_\text{left}$ to positions of some \infix $s_j\in(\exactcontent{C_j})^M$.
  We define $f(B_i,i)=(C_j,j)$ for the smallest such index $j$.

  The function $g : \mathbf{C} \rightarrow \mathbf{B}$ is defined
  symmetrically, by exchanging the roles of  $(v_n)_n$ and $(w_n)_n$.  
  Moreover, since $\content{r_i} = B_i$ and $r_i\preceq
  s_j\in(\exactcontent{C_j})^M$, we obtain that $f(B_i,i) = (C_j,j)$ entails $B_i \subseteq
  C_j$. Similarly, if $g(C_j, j) = (B_i,i)$, then $C_j \subseteq B_i$. 
  If we show that $f$ and $g$ define a bijective correspondence between $\mathbf{B}$ and $\mathbf{C}$, then $\ell=m$. The 
  above observations would then imply that $B_i = C_i$, for every $i$.

  To establish that $f$ and $g$ are inverses of one another, we apply Lemma 8.2.5 from \citep{JAbook}, which we shall first repeat:  
  \begin{lemma}[{\citeauthor{JAbook}, \citeyear[Lemma~8.2.5]{JAbook}}]
    \label{lemma-825}
    Let $X$ and $Y$ be finite totally ordered sets and let $P$ be a partially ordered
    set. Let $f : X \rightarrow Y, g : Y \rightarrow X, p : X
    \rightarrow P$ and $q : Y \rightarrow P$ be functions such that
    \begin{enumerate}
    \item\label{it:0} $f$ and $g$ are nondecreasing.
    \item\label{it:1} for any $x \in X, p(x) \leq q(f(x))$,
    \item\label{it:2} for any $y \in Y, q(y) \leq p(g(y))$,
    \item\label{it:3}if $x_1, x_2 \in X, f(x_1) = f (x_2)$ and $p(x_1)
      = q ( f(x_1))$, then $x_1 = x_2$,
    \item\label{it:4}if $y_1, y_2 \in Y, g(y_1) = g (y_2)$ and $q(y_1)
      = p ( g(y_1))$, then $y_1 = y_2$.
    \end{enumerate}
    Then $f$ and $g$ are mutually inverse functions and $p = q \circ f$ and $q = p \circ g$.     
  \end{lemma}
  We prove Lemma~\ref{lemma-825} for the sake of completeness, and because there
  is a minor mistake in \citep{JAbook}: the original statement lacks the hypothesis that
  $f,g$ are nondecreasing, and this hypothesis can easily be seen to be
  necessary.

  \begin{proofn}[of Lemma~\ref{lemma-825}]
    By symmetry, it suffices to prove that for all $x\in X$, we have
    $g(f(x))=x$. Note that this equality together with Item~\ref{it:1} and
    Item~\ref{it:2} applied to $y=f(x)$ then entails that $p=q\circ f$.

    Let $x\in X$, we have to show that $g(f(x))=x$. We define inductively two
    sequences, by $x_{0}=x$, $y_{n}=f(x_{n})\in Y$ and $x_{n+1}=g(y_{n})\in
    X$. Assume by contradiction that $x_{0}\not=x_{1}$. By symmetry, we may assume
    that $x_{0}<x_{1}$. We claim that the sequence $(x_{n})_{n}$ is strictly
    increasing, which immediately gives a contradiction since $X$ is finite.
    We prove this claim by induction. We already know that  $x_{0}<x_{1}$.

    Assume now that $x_{n-1}<x_{n}$, we show that $x_{n}<x_{n+1}$. Since $f$ is
    nondecreasing, we have $y_{n-1}\leq y_{n}$. We want to show that $y_{n-1}<
    y_{n}$. If on the contrary $y_{n-1}= y_{n}$, then
    $f(x_{n-1})= f(x_{n})$ and by Items~\ref{it:1} and \ref{it:2}, we would have,
    \[
      p(x_{n})\leq q(f(x_{n}))=q(y_{n})=q(y_{n-1})\leq p(g(y_{n-1}))=p(x_{n}),
    \]
    so that $p(x_{n})=q(f(x_{n}))$.  Item~\ref{it:3} applied to $x_{n}$ and
    $x_{n-1}$ would imply that $x_{n-1}=x_{n}$, contradicting our induction
    hypothesis. Therefore, $y_{n-1}<y_{n}$. With a similar argument, from  $y_{n-1}<y_{n}$ we
    obtain $x_{n}<x_{n+1}$, as desired.  \qed
  \end{proofn}
  
  To apply Lemma~\ref{lemma-825}, let $X = \mathbf{B}$ be ordered according to the
  second component, that is, $(B_{i},i)<(B_{j},j)$ when $i<j$.  Similarly, let
  $Y = \mathbf{C}$ be ordered according to the second component. Observe that the
  functions $f:\mathbf{B}\to\mathbf{C}$ and $g:\mathbf{C}\to\mathbf{B}$ defined
  above are nondecreasing by construction.  Finally, let $P$ be the set of
  subalphabets of $A$ partially ordered by inclusion, and let $p$ and $q$ be the
  projections onto the first coordinate.

  Let us verify that $f$ and $g$ fulfill all conditions of Lemma~\ref{lemma-825}.
  Item~\ref{it:1} holds since for all $i$ and $j$, $f(B_i,i) = (C_j,j)$
  implies that $B_i \subseteq C_j$.
  Item~\ref{it:2} holds symmetrically.

  For Item~\ref{it:3}, suppose that $f(B_{i_1},i_1) = f(B_{i_2},i_2)=(C_j,j)$ and
  that $p(B_{i_1},i_1)=q(f(B_{i_1},i_1))$, that is, $B_{i_1} =C_j$. The first
  condition, $f(B_{i_1},i_1) = f(B_{i_2},i_2)=(C_j,j)$, means that a \infix
  $r_{i_1}$ and a \infix $r_{i_2}$ of $v_1^{(k)}$ have all their positions
  mapped by $h_\text{left}$ to positions of $s_j$ in $w_{M}$. Therefore, all
  intermediate positions are also mapped to positions of $s_j$, which implies,
  using the second condition, $B_{i_1} =C_j$, that $\content{r_{i_1} u_{i_1} \cdots
    {r_{i_2}}} \subseteq \content{s_j} = C_j = B_{i_1} = \content{r_{i_1}}$. But we
  assumed that $(\overrightarrow{u},\overrightarrow{B})$ is a \emph{proper}
  factorization pattern, so $i_1$ must be equal to $i_2$. This shows that
  Item~\ref{it:3} holds. Item~\ref{it:4} is dual.

  It follows that indeed $f$ and $g$ define a bijective correspondence
  between $\mathbf{B}$ and $\mathbf{C}$, thus  $\ell=m$ and $B_i = C_i$,
  for every $i$. Since we are dealing with proper factorization
  patterns, $v_1^{(k)} \preceq w_{M}$ now implies that $u_i \preceq t_i$ for
  every $i$. Symmetrically, $t_i \preceq u_i$ for every $i$. Thus, for every $i$, $u_i = t_i$.
\end{proof}

\noindent
The proof of Theorem~\ref{theo:patterns} relies on Lemmas~\ref{lem:extract-adequate-subseq} and \ref{lem:samepatt}, as well as on the following characterization of \ptl-(non)-separation.

\begin{lemma}\label{lem:sepcarac}
  Let $I$ and $E$ be two languages. Then, $I$ and $E$ are \emph{not} \ptl-separable if and only
  if for every $n \in \mathbb{N}$, there exist $v_n \in I$ and $w_n \in E$ such that $v_n \sim_n w_n$.
\end{lemma}

\begin{proof}
  We first prove the implication from left to right by contraposition: suppose that there
  exists $n \in \mathbb{N}$ such that $(I\times E)\cap{\sim_{n}}=\emptyset$. We have to prove that
  $I$ and $E$ are \ptl-separable. The assumption implies that
  $[I]_n := \{v \mid \exists w \in I$ such that $v\sim_n w\}$ separates $I$ from $E$. It remains to verify
  that it is a \ptl.  Observe that $\sim_n$ has finite index, because an
  $\sim_n$-class is defined by a subset of words of length at most $n$.  Therefore, $[I]_n$ is a
  finite union of $\sim_n$-classes. To show that $[I]_n$ is a \ptl, it is therefore enough to check
  that a single $\sim_n$-class is a \ptl, which is the case since the $\sim_n$-class of $u$ is
  \[
    \bigcap_{v\preceq u,\ |v|\leq n}v{\uparrow} \cap \bigcap_{v\not\preceq u,\ |v|\leq
      n}A^*\setminus (v{\uparrow})
  \]
  where $v{\uparrow}$ is the $A$-upward closure of $\{v\}$, \emph{i.e.}, the piece language
  $A^*a_1A^*\cdots A^*a_kA^*$, if $v=a_1\cdots a_k$.

  For the other direction, we first note that a piecewise testable language is a union of $\sim_n$-equivalence classes for some $n$.
  Indeed, notice that by definition, a piecewise testable language is a finite Boolean combination of piece languages. Let $n$ be the maximal length of the pieces defining these piece languages.
  Any piece language $A^*a_1A^*\cdots A^*a_kA^*$, may, for $v=a_1\cdots a_k$, be written as $\bigcup_{v \preceq w} [w]_{\sim_k}$. This is a finite union, as $\sim_k$ has finite index. For every $k$, the equivalence relation $\sim_{k+1}$ is a refinement of $\sim_k$. Therefore, $A^*a_1A^*\cdots A^*a_kA^*$ may also be written as a finite union of $\sim_n$-equivalence classes. And, clearly, a finite Boolean combination of finite unions of $\sim_n$-equivalence classes is again a finite union of $\sim_n$-equivalence classes.
  Now suppose that for every $n \in \mathbb{N}$, there are words in $I$ (respectively, $E$) that are $\sim_n$-equivalent. Since a piecewise testable language is a union of $\sim_n$-equivalence classes for some $n$, any piecewise testable language containing $I$ will, in this case, have nonempty intersection with $E$.
\end{proof}

We are now able to prove Theorem~\ref{theo:patterns} using Lemmas~\ref{lem:extract-adequate-subseq}, \ref{lem:samepatt}
and~\ref{lem:sepcarac}.

\smallskip
\noindent \textbf{Proof of Theorem~\ref{theo:patterns}.}
{\it
Two word languages $I$ and $E$ are not separable by \ptl if and only
  if they contain a common pattern $(\overrightarrow{u},
  \overrightarrow{B})$.}

\begin{proof}
  To prove the ``if''-direction of the theorem, we show that for every $n$, there
  exist words in $I$ and in~$E$ that are $\sim_n$-equivalent.  By hypothesis, $I$
  and $E$ contain a common factorization pattern
  $(\overrightarrow{u},\overrightarrow{B})$. By definition, this gives us two sequences
  of words $(v_n)_n$ and $(w_n)_n$ such that for all $n$,
  \begin{align*}
    & v_n\in I\cap \LangExp  u B n p,\\
    & w_n\in E\cap \LangExp  u B n p.
  \end{align*}
  Observe that for all $i$, $(\exactcontent{B_i})^n$ contains precisely all
  words from $B_i^{\leq n}$ as \subwords of size up to~$n$. It follows that $v_n \sim_n w_n$.

  We now prove the ``only-if'' direction. 
  The existence of $v_n \sim_n w_n$ for every $n \in \mathbb{N}$ defines an infinite sequence of pairs $(v_n,w_n)_n$, from which
  we will iteratively extract infinite subsequences to obtain
  additional properties, while
  keeping~$\sim_n$-equivalence. 

  By Lemma~\ref{lem:extract-adequate-subseq}, one can extract from $(v_n,w_n)_n$ a subsequence whose first component forms an adequate sequence. From this subsequence of pairs, using Lemma~\ref{lem:extract-adequate-subseq} again, we extract a subsequence whose second component is also adequate (note that the first component remains adequate). Therefore, one can assume that both $(v_n)_n$ and $(w_n)_n$ are themselves adequate. This means there exist proper factorization patterns for which $(v_n)_n$ resp.~$(w_n)_n$ are adequate.
  %
%
%
  Since $v_n\sim_n w_n$, Lemma~\ref{lem:samepatt} shows that one can choose the \emph{same}
  proper factorization pattern $(\overrightarrow{u},\overrightarrow{B})$ such that both
  $(v_n)_n$ and $(w_n)_n$ are $(\overrightarrow{u},\overrightarrow{B})$-adequate. This means
  that $I$ and $E$ contain a common pattern
  $(\overrightarrow{u},\overrightarrow{B})$. 
 \end{proof}

\newcommand{\impl}[2]{``$\ref{#1}\Rightarrow\ref{#2}$''}

\section{The Characterization for Separability}\label{sec:separability}

In this section we prove our main characterization:

\medskip \noindent \textbf{Theorem~\ref{theo:decidability}.}
{\it
 For each full trio $\C$, the following are equivalent:
  \begin{enumerate}[label=(\arabic*), ref=(\arabic*)]
  \item Separability of $\C$ by \ptl is decidable.
  \item The diagonal problem for $\C$ is decidable.
  \item The SUP for $\C$ is decidable.
  \item Downward closures are computable for $\C$.
  \end{enumerate}}
The equivalence between \ref{theo:it:sup} and \ref{theo:it:dc} is
immediate from Theorem~\ref{theo:dc}. The implication
\impl{theo:it:diag}{theo:it:sup} is trivial because the SUP is a
special case of the diagonal problem. We prove the other direction
\impl{theo:it:sup}{theo:it:diag} and the implication \impl
{theo:it:sep}{theo:it:sup} in Section~\ref{sec:diagonal-algo}.
Finally, we prove the implication \impl{theo:it:diag}{theo:it:sep} in
Section~\ref{sec:separability-algo}, by giving an algorithm for
separability if the diagonal problem is decidable.

\subsection{Algorithms for the Diagonal Problem and the SUP}\label{sec:diagonal-algo}

In this section, we reduce the diagonal problem to the simultaneous
unboundedness problem and to separability by \ptl.  It should be noted
that the implication \impl{theo:it:sup}{theo:it:diag} already follows
easily from the equivalence between \ref{theo:it:sup} and
\ref{theo:it:dc} (shown in \citep{Zetzsche-icalp15}): when one can
compute downward closures, the diagonal problem reduces to the case of
regular languages. However, the algorithm for downward closure
computation in \citep{Zetzsche-icalp15} is not accompanied by any
complexity bounds.  Therefore, we mention here a reduction that can be
carried out in non-deterministic polynomial time if we assume that the
full trio operations require only polynomial time
(Lemma~\ref{lem:sup-to-diagonal}).  This reduction is based on the
fact that downward closed languages can be written as finite unions of
ideals, which is also a central ingredient in
\citep{Zetzsche-icalp15}. An \emph{ideal} is a language of the form
$B_0^*\{b_1,\varepsilon\}B_1^*\cdots \{b_m,\varepsilon\}B_m^*$, where
$m\ge 0$, $b_1,\ldots,b_m$ are symbols, and $B_0,\ldots,B_m$ are
alphabets. Clearly, every ideal is downward closed. We use the
following result by Jullien, which has later been rediscovered
in \citep{Abdulla2004}.

\begin{theorem}[\citeauthor{Jullien1969}, \citeyear{Jullien1969}]\label{theo:ideals}
  Every downward closed language is a finite union of ideals.
\end{theorem}
In fact, there is a notion of ideal in well quasi orderings for which this result holds
in general. The property of the theorem is even equivalent to the fact that the
underlying quasi ordering has no infinite antichains, see
\citep{bonnet75,Fraisse:Theory-Relations} and \citep{ErdosTarski-mutually43}.  See
also~\citep{FGL09,LSdemyst15} for other perspectives.

\smallskip
We now prove the implication \impl{theo:it:sup}{theo:it:diag} of
Theorem~\ref{theo:decidability}.
In the following lemma, we use the definition of full trios as being closed
under rational transductions.  A subset $T\subseteq A^*\times B^*$ for
alphabets $A,B$ is called a \emph{rational transduction} if there is an alphabet
$C$, a regular language $R\subseteq C^*$, and morphisms $\alpha\colon
C^*\to A^*$, $\beta\colon C^*\to B^*$ such that we have $T=\{(\alpha(w), \beta(w)) \mid
w\in R\}$. For a language $L\subseteq A^*$ and a rational transduction
$T\subseteq A^*\times B^*$, we define
\[ TL=\{w\in B^* \mid \exists v\in L\colon (v,w)\in T\}. \]
It is well-known that a nonempty class of languages $\C$ is an effective full trio if and only if for
each $L\in\C$ and each rational transduction $T$, we have effectively
$TL\in\C$, see~\citep{Berstel:Transductions-context-free-languages:1979:a} for instance.

Observe that the set 
\[ D=\{(u,v)\in A^*\times A^* \mid v\preceq u\} \] is a rational
transduction, meaning that for every member $L$ of a full trio $\C$,
the language $\Dclosure{L}=DL$ is effectively contained in $\C$ as
well.

\begin{lemma}\label{lem:sup-to-diagonal}
  Let $\C$ be a full trio. If the SUP is decidable for $\C$, then so is the diagonal
  problem for~$\C$.
\end{lemma}

\begin{proof}
  Let $L\subseteq A^*$ with $A=\{a_1,\ldots,a_n\}$. We say that $L$
  satisfies the \emph{diagonal property} if for each $m\in\N$, the
  vector $(m,\ldots,m)$ (with $|A|$ entries) is dominated by some
  vector in the Parikh image of~$L$.  We claim that $L$ satisfies the
  diagonal property if and only if we can order the symbols of $A$ as
  $A=\{b_1\le\cdots\le b_n\}$ such that $b_1^*\cdots b_n^*\subseteq
  \Dclosure{L}$. The ``if'' direction is immediate, so suppose $L$
  satisfies the diagonal property.

  Since $L$ satisfies the diagonal property, so does
  $\Dclosure{L}$.  According to Theorem~\ref{theo:ideals}, we can
  write $\Dclosure{L}$ as a finite union of ideals. This means that
  at least one of these ideals has to satisfy the diagonal property. Let
  $I=C_0^*\{c_1,\varepsilon\}C_1^*\cdots \{c_m,\varepsilon\}C_m^*$ be
  one such ideal. Then we have $C_0\cup \cdots\cup C_m=A$, since any
  element of $A\setminus (C_0\cup \cdots\cup C_m)$ can occur at most
  $m$ times in words in $I$. Hence, if we order the elements of $A$
  linearly by picking first the elements of $C_0$, then those of
  $C_1\setminus C_0$, and so forth, the resulting ordering $A=\{b_1\le \cdots\le
  b_n\}$ clearly satisfies $b_1^*\cdots b_n^*\subseteq I\subseteq
  \Dclosure{L}$. This proves our claim.

  Assume that $\C$ is a full trio for which the SUP is decidable. We
  show that the diagonal problem is decidable for $\C$ as well. Let
  $L\subseteq A^*$ be a language. .

  We guess a linear ordering $A=\{b_1\le \cdots \le b_n\}$. Notice
  that since $\C$ is closed under intersection with regular sets, the
  language $K=\Dclosure{L}\cap b_1^*\cdots b_n^*$ is effectively
  contained in $\C$.  Then, $K$ is an instance of the SUP. By our
  claim, $L$ satisfies the diagonal property if and only if there
  exists a linear ordering $A=\{b_1\le \cdots \le b_n\}$ such that
  $K=\Dclosure{L}\cap b_1^*\cdots b_n^*$ is a positive instance of the
  SUP.
\end{proof}

This concludes the proof of \impl{theo:it:sup}{theo:it:diag} of
Theorem~\ref{theo:decidability}, so that we now have established that Items
\ref{theo:it:diag}, \ref{theo:it:sup} and \ref{theo:it:dc} are equivalent. It
remains to connect these problems with the separation problem.

The correctness proof of our reduction of the SUP to the separability problem
employs a characterization of separability by \ptl from
\citep{DBLP:conf/icalp/CzerwinskiMM13}.  For languages $K,L\subseteq A^*$, a
\emph{$(K,L)$-zigzag} is an infinite sequence of words $(w_i)_{i\in\N}$ such that
\begin{compactenum}[(i)] 
\item $w_i\preceq w_{i+1}$ for every $i\in\N$, 
\item $w_i\in K$ for every even $i$ and 
\item $w_i\in L$ for every odd $i$.
\end{compactenum}

\smallskip
The characterization for separability by \ptl is then as follows.

\begin{theorem}[\citeauthor{DBLP:conf/icalp/CzerwinskiMM13},
  \citeyear{DBLP:conf/icalp/CzerwinskiMM13}]
  \label{thm:zigzag}
Let $K,L\subseteq A^*$ be languages. Then $L$ and $K$ are separable by \ptl
if and only if there is no $(K,L)$-zigzag.
\end{theorem}

\begin{lemma}
  Let $\C$ be a full trio. If separability by \ptl is decidable for $\C$, then
  so is the SUP for $\C$.
\end{lemma}

\begin{proof}
  Let $A=\{a_1,\ldots,a_n\}$ and let $L\subseteq a_1^*\cdots a_n^*$ be
  a language in $\C$. Let $T\subseteq A^*\times A^*$ be the rational
  transduction
  \[ T=\{ ( a_1^{k_1}\cdots a_n^{k_n}, a_1^{2k_1}\cdots a_n^{2k_n})
  \mid k_1,\ldots,k_n\ge 0\}. \] Moreover, consider the regular
  language
  \[ K=\{ a_1^{2k_1+1}\cdots a_n^{2k_n+1} \mid k_1,\ldots,k_n\ge
  0\}. \] 

  We claim that $T(\Dclosure{L})$ and $K$ are inseparable by \ptl if
  and only if $\Dclosure{L}=a_1^*\cdots a_n^*$. We first show that
  this claim implies the lemma. Indeed, since $\C$ is a full trio, the
  language $T(\Dclosure{L})$ is effectively contained in $\C$. Since
  $K$ is a member of $\C$ (since every full trio contains the family
  of regular languages), the claim clearly implies the lemma.

  We now prove the claim. Suppose that $\Dclosure{L}=a_1^*\cdots a_n^*$. Then we have
  $T(\Dclosure{L})=(a_1 a_1)^*\cdots (a_n a_n)^*$ and the sequence
  $(w_i)_{i\in\N}$ with $w_i=a_1^i\cdots a_n^i$ is a
  $(T(\Dclosure{L}), K)$-zigzag, meaning that $T(\Dclosure{L})$ and $K$ are
  inseparable by \ptl, by Theorem~\ref{thm:zigzag}.

  Now suppose $T(\Dclosure{L})$ and $K$ are inseparable by \ptl. Then again by
  Theorem~\ref{thm:zigzag}, 
  there is a $(T(\Dclosure{L}), K)$-zigzag $(w_i)_{i\in\N}$. This
  means in particular that, for all $i$, we have $w_i\preceq w_{i+1}$, so
  that $|w_i|_a\le|w_{i+1}|_a$ for $a\in A$.  By construction of $T$
  and $K$, the numbers $|w_i|_{a}$ and $|w_{i+1}|_{a}$ are incongruent
  modulo $2$, which implies $|w_i|_a<|w_{i+1}|_a$ and therefore
  $|w_i|_a\ge i$. Moreover, since the sequence is a zigzag, we have
  $\{w_{2i} \mid i\ge 0\}\subseteq T(\Dclosure{L})$ and thus
  $\Dclosure{L}=a_1^*\cdots a_n^*$. This completes the proof of our
  lemma.
\end{proof}

\subsection{The Algorithm for Separability}\label{sec:separability-algo}

It only remains to prove the implication
\impl{theo:it:diag}{theo:it:sep} to finish the proof of
Theorem~\ref{theo:decidability}. In this section we show
that, for full trios with decidable diagonal problem, we can decide
separability by \ptl.  Fix two
languages $I$ and $E$ from a full trio $\C$ which has decidable diagonal problem.

To test whether $I$ is separable from $E$ 
by a piecewise testable language $S$, we run two
semi-procedures in parallel.  The \emph{positive} one looks for a witness that $I$
and $E$ are separable by \PTL, whereas the \emph{negative} one looks
for a witness that they are \emph{not} separable by a \PTL.
Since one of the semi-procedures always terminates, we have an
effective algorithm that decides separability. It remains to describe
the two semi-procedures.

\segment{Positive semi-procedure.}  We first note that, when a full
trio has decidable diagonal problem, it also has decidable emptiness.
Indeed, emptiness of a language $L\subseteq A^*$ can be decided
by taking the $A$-upward closure of $L$ (which can be effectively computed
from $L$, since it can be implemented by a rational transduction). In the resulting language,
the diagonal problem returns true if and only if $L$ is nonempty.

The positive semi-procedure enumerates all \PTL{}s over the union of
the alphabets of $I$ and $E$. For every \ptl $S$ it checks whether $S$
is a separator, so if $I \subseteq S$ and $E \, \cap \, S =
\emptyset$.  The first test is equivalent to $I \, \cap \, (A^*
\setminus S) = \emptyset$.  Thus both tests boil down to checking
whether the intersection of a language from the class $\C$ ($I$ or
$E$, respectively) and a regular language ($S$ and $A^* \setminus S$,
respectively) is empty. This is decidable, as $\C$ is effectively
closed under taking intersections with regular languages and has
decidable emptiness problem.

\segment{Negative semi-procedure.}
Theorem~\ref{theo:patterns} shows that there is always a finite witness
for inseparability: a pattern $(\overrightarrow{u}, \overrightarrow{B})$.
The negative semi-procedure enumerates all possible patterns and for
each one, checks the condition of 
Theorem~\ref{theo:patterns}.  We now show how to test this condition, \emph{i.e.}, for a pattern
$(\overrightarrow{u}, \overrightarrow{B})$ test whether \emph{for all} $n \in \N$ the intersection
of $\L(\overrightarrow{u}, \overrightarrow{B}, n)$ with both $I$ and $E$ is nonempty. Note that
the difficulty for testing this condition comes from the universal quantification over $n$.

\medskip
\noindent \textit{Checking the condition.}
Here we show for an arbitrary language from $\C$ how to check
whether for all $n \in \N$ its intersection with the
language $\L(\overrightarrow{u}, \overrightarrow{B}, n)$ is nonempty.
Fix $L \in \C$ over an alphabet $A$
and a pattern $(\overrightarrow{u}, \overrightarrow{B})$,
where $\overrightarrow{u} = (u_0, \ldots, u_k)$ and
$\overrightarrow{B} = (B_1, \ldots, B_k)$.
Intuitively, we just consider a diagonal problem with some artifacts:
we are counting the number of ``full occurrences'' of alphabets $B_i$ and checking whether
those numbers can simultaneously become arbitrarily big.

We show decidability of the non-separability problem by a formal
reduction to the diagonal problem. We perform a sequence of steps. In
every step we will slightly modify the considered language $L$ and
appropriately customize the condition to be checked.  Using the
closure properties of the full trio $\C$ we will ensure that the
investigated language still belongs to $\C$.

First we add special symbols $\$_i$, for $i \in 
\{1,\ldots,k\}$, which do not occur in $A$. These symbols are meant to count
how many times alphabet $B_i$ is ``fully occurring'' in the word. 
Then we will assure that words are of the form
\begin{equation*}
  u_0 \, (B_1 \cup \{\$_1\})^* \, u_1 \cdots u_{k-1} \, (B_k \cup \{\$_k\})^* \, u_k,
\end{equation*}

\noindent which already is close to what we need for the pattern.
Then we will check that between every two symbols $\$_i$ (with the
same $i$), every symbol from $B_i$ occurs,
so that the $\$_i$ are indeed counting the number of iterations
through the entire alphabet $B_i$.
Finally we will remove all the symbols except those from $\{\$_1,\ldots,\$_k\}$. The resulting language
will contain only words of the form $\$_1^* \$_2^* \cdots \$_k^*$ and the condition to be checked
will be exactly the diagonal problem.

More formally, let $L_0 := L$. We modify iteratively $L_0$, resulting in $L_1,
L_2, L_3$, and $L_4$. Each of them will be in $\C$ and we describe them next.

Language $L_1$ is the $\{\$_1, \ldots, \$_k\}$-upward closure of
$L_0$.  Thus, $L_1$ contains, in particular, all words
where the $\$_i$ are placed ``correctly'', \emph{i.e.}, in between two $\$_i$-symbols
the whole alphabet $B_i$ should occur.  However at this moment we do
not check it.  By closure under $B$-upward closures, which can be implemented
by rational transductions, language $L_1$
belongs to $\C$.

Note that $L_1$ also contains words in which the $\$_i$-symbols are
placed totally arbitrary. In particular, they can occur in the wrong
order. 
The idea behind $L_2$ is to consider only those words in which the $\$_i$-symbols
are placed at least in the ``right areas''.
Concretely, $L_2$ is the intersection of $L_1$ with the language
\begin{equation*}
  u_0 \, (B_1 \cup \{\$_1\})^* \, u_1 \cdots u_{k-1} \, (B_k \cup \{\$_k\})^* \, u_k.
\end{equation*}

\noindent
Since $\C$ is a full trio, it is  closed under intersection with regular
languages, whence $L_2$ belongs to $\C$.

Language $L_2$ may still contain words, such that in between two $\$_i$-symbols
not \emph{all} the symbols from $B_i$ occur. We get rid of these by
intersecting $L_2$ with the regular language
\begin{equation*}
  u_0  (\$_1 B_1^{\circledast})^* \$_1 u_1 \cdots u_{k-1} (\$_kB_k^{\circledast})^* \$_k u_k.
\end{equation*}
\noindent As such, we obtain $L_3$ which, again by
closure under intersection with regular languages, belongs to
$\C$.\footnote{Of course, one could also immediately obtain $L_3$ from
  $L_1$ by performing a single intersection with a regular language.}

Note that the intersection of $L = L_0$ with the
language $\L(\overrightarrow{u}, \overrightarrow{B}, n)$ is nonempty if and only if $L_3$
contains a word with precisely $n+1$ symbols $\$_i$ for every $i \in \{1, \ldots, k\}$. Indeed, $L_3$
just contains the (slightly modified versions of) words from $L_0$
which fit into the pattern and in which the symbols $\$_i$
``count'' occurrences of $B_i^\circledast$. Furthermore, for every word in
$L_3$, the word obtained by removing some occurrences of some $\$_i$
is in $L_3$ as well. It is thus enough to focus on the $\$_i$-symbols.
Language $L_4$ is therefore the $\{\$_1, \ldots, \$_k\}$-projection of $L_3$.
By closure under projections, language $L_4$ belongs to $\C$.
The words contained in $L_4$ are therefore of the form
\begin{equation*}
\$_1^{a_1} \, \cdots \, \$_k^{a_k},
\end{equation*}

\noindent such that there exists $w \in L$ with at least $a_i-1$ occurrences of
\infixes of $B_i^\circledast$ between $u_{i-1}$ and $u_i$.  Therefore,
$L\cap \L(\overrightarrow{u}, \overrightarrow{B}, n)$ is nonempty for all $n
\geq 0$ if and only if the tuple $(n, \ldots, n)$ is dominated by an element of the
Parikh image of $L_4$ for infinitely many $n \geq 0$.  This is
precisely the diagonal problem, which we know to be decidable for
$\C$.\qed







\section{Decidable Classes}\label{sec:feasibility}

In this section we show that separability by piecewise testable languages is decidable for a wide range of classes, by proving that
they meet the conditions of Theorem~\ref{theo:decidability}, in
particular, for context-free languages and languages of
labeled vector addition systems (which are the same as languages of
labeled Petri~nets). 

\segment{Effectively semilinear Parikh images.}
A number of language classes is known to exhibit effectively
semilinear Parikh images.  A set $S \subseteq \N^k$ is \emph{linear}
if it is of the form $$S = \{v + n_1 v_1 + \ldots + n_m v_m \mid n_1,
\ldots, n_m \in \N \}$$ \noindent for some \emph{base} vector $v \in
\N^k$ and \emph{period} vectors $v_1, \ldots, v_m \in \N^k$. A
\emph{semilinear} set is a finite union of linear sets.
We say that a full trio $\C$ exhibits \emph{effectively semilinear
Parikh images} if every language in $\C$ has a semilinear Parikh image
and one can compute a representation as a (finite) union of linear sets. 

Clearly, one can decide the diagonal problem for language classes with
effectively semilinear Parikh images. This amounts to checking whether
in a representation of the Parikh image of the given language, there
is some linear set in which each of the symbols occurs in some period
vector. 

Another option is to use Presburger logic.  Semilinear sets are
exactly those definable by Presburger logic.  Moreover, the
translation is effective.  Assume that $|A| = k$, so the
Parikh image $P$ of the considered language is a subset of $\N^k$ and
$\phi$ is a Presburger formula describing $P$ having $k$ free
variables.  Then
\[
\psi = \forall_{n \in \N}\ \exists_{x_1, x_2, \ldots, x_k}\ 
\big( \bigwedge_{i \in \{1, \ldots, k\}} (x_i \geq n) \big) \wedge 
\phi(x_1, x_2, \ldots, x_k)
\]
is true if and only if the diagonal problem for the considered
language is answered positively.  Decidability of the Presburger logic
finishes the proof of decidability of the diagonal problem.
We refer for the details of semilinear sets and Presburger
logic to~\citep{Ginsburg66}.

Examples of full trios with effectively semilinear Parikh images are
\emph{context-free languages}~\citep{DBLP:journals/jacm/Parikh66},
\emph{multiple context-free languages}~\citep{Seki1991}, languages of
\emph{reversal-bounded counter automata}~\citep{Ibarra1978}, \emph{stacked
counter automata}~\citep{Zetzsche2015a} and \emph{finite index matrix
languages}~\citep{DassowPaun1989}.

\segment{Higher-Order Pushdown Automata and Recursion Schemes.} Very
recently, it was shown that the diagonal problem is decidable for
higher-order pushdown automata~\citep{HagueKochemsOng2016} and even
higher-order recursion
schemes~\citep{ClementeParysSalvatiWalukiewicz2016}. Both of these
methods use an inductive approach: The diagonal problem for
order-$(n+1)$ pushdown automata (resp. schemes) is decided by
constructing an order-$n$ pushdown automaton (resp. scheme) that is
equivalent with respect to the diagonal problem.  Hence, the algorithm
arrives at order 0, for which the diagonal problem is easy to solve.

\segment{Languages of Labeled Vector Addition Systems and Petri Nets.}
A $k$-dimensional \emph{(labeled) vector addition system}, or
\emph{(labeled) VAS} $M = (A, T, \delta_0, \delta_1, \ell, s, t)$ over
alphabet $A$ consists of a finite set of \emph{transitions} $T$, a
labeling $\ell: T \to A \cup \{\eps\}$, mappings
$\delta_0,\delta_1:T\to\N^k$, and \emph{source} and \emph{target}
vectors $s, t \in \N^k$.  A labeled VAS defines a transition relation
on the set $\N^k$ of \emph{markings}. For two markings $u, v \in \N^k$
we write $u \trans{a} v$ if there is an $r \in T$ such that
$\delta_0(r)\le u$ and $v=u-\delta_0(r) + \delta_1(r)$ and $\ell(r) =
a$, where the addition and comparison of vectors is defined coordinate-wise.  For two
markings $u, v \in \N^k$ we say that \emph{$u$ reaches $v$ (via the
  word $w$)} if there is a sequence of markings $u_0 = u, u_1, \ldots,
u_{n-1}, u_n = v$ such that $u_i \trans{a_i} u_{i+1}$ for all $i \in
\{0, \ldots, n-1\}$ and $w = a_0 \cdots a_{n-1}$.  For a given labeled
VAS $M$ the \emph{language} of $M$, denoted $L(M)$, is the set of all
words $w \in A^*$ such that source reaches target via $w$.  We note
that languages of labeled VASs are the same as languages of labeled
Petri nets.

Since labeled VAS languages are known to be a full
trio~\citep{Jantzen-rairo79}, we only need to prove decidability of
the diagonal problem.  For instance, this can be done by using the
computability of downward closures for VAS
languages~\citep{HabermehlMeyerWimmel2010}.

\smallskip
Here, we present an alternative approach to this problem.  We show
that the \emph{diagonal problem for VAS languages} is decidable by reduction to the
\emph{place-boundedness problem for VASs with one zero test}, which has been
shown decidable~\citep{vaszero10,DBLP:journals/corr/abs-1205-4458}.  A
$k$-dimensional \emph{(labeled) vector addition system with one
  zero-test}, or \emph{(labeled) VAS${}_z$} over the alphabet $A$ is a
tuple $M = (A, T, r_z, \delta_0, \delta_1, \ell, s, t)$ such that
$(A,T,\delta_0,\delta_1,\ell,s,t)$ is a VAS and additionally $r_z\in T$
is a distinguished \emph{zero-test transition}. The semantics of a
VAS${}_z$ differs slightly from that of VASs.  For two markings $u, v \in \N^k$ in
a VAS${}_z$ with $u=(u_1,\ldots,u_k)$, we write $u \trans{a} v$ if
there is $r \in T$ such that \begin{enumerate*}[(i)] \item
  $\delta_0(r)\le u$, \item $v=u - \delta_0(r)+\delta_1(r)$, \item if
  $r=r_z$, then $u_1=0$, and \item $\ell(r) = a$.\end{enumerate*}
Then, reaching another marking via a word is defined accordingly as in
a VAS.  In \citep{vaszero10,DBLP:journals/corr/abs-1205-4458}, it was shown that
given a $k$-dimensional VAS${}_z$ and an $i\in\{1,\ldots,k\}$, it is
decidable whether for every $n\in\N$, there is a reachable marking
$u=(u_1,\ldots,u_k)$ with $u_i\ge n$.  This is known as the
\emph{place-boundedness problem}.

\smallskip To set up our reduction, we start from a VAS $M$. Our objective
is to build from $M$ a VAS$_z$ $M_z$, such that $M_z$ is unbounded in
some distinguished coordinate if and only if $L(M)$ has the diagonal
property. To build $M_z$ from $M$, we proceed in two steps.

\begin{description}[leftmargin=*]
\item [Step 1] The first step consists in modifying the VAS~$M$ to get another
  VAS $M'$, in the following
  way.
  \begin{itemize}[label=--,topsep=0pt,leftmargin=*]
  \item First, it is well-known that any VAS can be turned into one
    generating the same language and having $(0,\ldots,0)\in\N^k$ as target
    marking~\citep{Hack1976}, meaning we may assume that the target
    is zero.

  \item Second, we introduce a new \emph{sum coordinate}, say as
    coordinate $1$: it is easy to modify the VAS so that in the new
    first coordinate, one always has the sum of all original coordinates
    of $M$.

  \item Finally, we add for each $a\in A$ a \emph{symbol coordinate}, which
    counts how many times we read symbol~$a$.  That is, for every
    transition which is labeled by $a \in A$, we add $1$ in the
    symbol-coordinate corresponding to $a$ and $0$ in the
    symbol-coordinates corresponding to other symbols. Hence, the set of
    symbol-coordinates always contains the Parikh image of the read
    prefix.
  \end{itemize}

\item [Step 2]
  In the second step, we turn the VAS $M'$ into a VAS${}_z$ $M_z$. For
  this, we add,
  \begin{itemize}[label=--,topsep=0pt,leftmargin=*]
  \item two \emph{mode coordinates}: the VAS$_z$ $M_z$ will be able to switch
    from the first mode to the second mode (but not the other way around), and

  \item a \emph{minimum coordinate}, counting the minimum of all
    \emph{original} coordinates. The intention is that this specific coordinate will be unbounded
    if and only if the original VAS $M$ has the diagonal property.
  \end{itemize}

  More precisely:
  \begin{enumerate}
  \item The initial marking is enriched by $1$ on the first mode
    coordinate and by $0$ in the second one.

  \item All transitions of $M'$ are changed so that they subtract $1$
    from the first mode coordinate and add~1 to it as well (\emph{i.e.},
    on the first mode coordinate, $\delta_0(r)$ is $-1$ and
    $\delta_1(r)$ is $1$). This means that the resulting transitions do not modify
    the mode coordinates, but can be fired only in the first mode.

  \item The new \emph{zero-test transition} $r_z$ tests to 0 the sum
    coordinate that was introduced in Step 1. Moreover, it only subtracts $1$
    from the first mode coordinate and adds $1$ to the second mode
    coordinate. In other words, it is applicable only in the first mode,
    and switches to the second mode.

  \item Finally, we add a new transition $\bar{r}$, labeled by
    $\varepsilon$, which subtracts $1$ from each symbol coordinate and adds
    $1$ to the minimum coordinate. Moreover, $\bar{r}$ subtracts $1$
    from the second mode coordinate and adds $1$ to it (in other words,
    it does not affect the mode, but applies only in the second~one).
  \end{enumerate}

\end{description}
Let us argue that the resulting VAS${}_z$ $M_z$ is unbounded in the minimum
coordinate if and only if $L(M)$ has the diagonal property. If $L(M)$
has the diagonal property, then for each $n$, we find a run in $M$
reading a word $w$ whose Parikh image dominates $(n,\ldots,n)$. Hence,
we can take the corresponding run in the VAS${}_z$, then fire $r_z$
once and then fire $\bar{r}$ exactly $n$ times to get the minimum
coordinate to $n$.

\smallskip On the other hand, if the minimum coordinate is
unbounded in the VAS${}_z$ $M_z$, then for each $n\ge 0$, there is a run in
which the minimum coordinate holds at least $n$. The minimum
coordinate can only have a nonzero content if the zero-test
transition had been applied beforehand, so that this run has to
correspond to a run in $M$ that arrives at $(0,\ldots,0)\in\N^k$ (since the
zero-test is performed on the sum coordinate), and
reads a word whose Parikh image dominates $(n,\ldots,n)$. We have
therefore reduced the diagonal problem for VAS languages to the
place-unboundedness problem for VAS${}_z$. The former is therefore
decidable.


\segment{Effective Parikh equivalence.} Let us mention another obvious
consequence of our characterization. Assume separability by \ptl is
decidable for a full trio $\D$. Furthermore, suppose that for each
given language~$L$ in a full trio $\C$, one can effectively produce a
Parikh equivalent language in $\D$.  Then, separability by \ptl is
also decidable for $\C$. For example, this means that separability by
\ptl is decidable for \emph{matrix languages}~\citep{DassowPaun1989}, a
natural language class that generalizes VAS languages and context-free
languages. Indeed, it is well-known that given a matrix language, one can
construct a Parikh equivalent VAS language~\citep{DassowPaun1989}.

\segment{Mixed instances.} Our result further yields that if
separability by \ptl is decidable for each of the full trios $\C$ and
$\D$, then it is also decidable whether given $K\in\C$ and $L\in\D$
are separable by a \ptl.
Indeed, if the diagonal problem is decidable for $\C$ and for $\D$,
then combining the respective algorithms yields decidability of the
diagonal problem for the full trio $\C\cup\D$.  For instance,
separability of a context-free language from the language of a labeled
vector addition system is decidable.
\begin{corollary}
  Let $\C$ and $\D$ be full trios. If separability by \ptl is
  decidable for $\C$ and for $\D$, then it is also decidable for $\C
  \cup \D$.
\end{corollary}

\section{Undecidable Classes}
\label{sec:undecidable-classes}

Given the fact that separability by \ptl is decidable for context-free
languages, the question arises whether the same is true of the
context-sensitive languages. Here we show that this is \emph{not} the case,
because of the following
reasons. 
First of all, context-sensitive languages are closed under
complementation. Therefore, decidable separability of
context-sensitive languages by \ptl would imply that it is decidable
to test if a given context-sensitive language is a \ptl.  The latter
question, however, is already undecidable for context-free languages
(Corollary~\ref{cor:cfl-testing}). In this section we provide these
observations with more detail and also exhibit a more general
technique to show undecidable separability by \ptl.


Thomas Place pointed out to us that, by a
proof that follows similar lines to the proof of Greibach's
Theorem~\citep{Greibach-mst68}, one can show that testing whether a
given context-free language is piecewise testable is undecidable
\citep{Place-personal15}. We give the proof below
(Proposition~\ref{prop:Lab-undecidable}).

We noticed that the proof he gave us even applies to all full trios
that contain the language,
\[
  L_{ab} = \{a^n b^n \mid n \geq 0\}.
\]
Note that this means that for \emph{all the language classes mentioned
  in Section~\ref{sec:feasibility}}, for which \ptl-\emph{separability} is
\emph{decidable}, it is \emph{undecidable} whether a given language
\emph{is} a \ptl.

We introduce some additional material to allow for this
generalization.  Given a language $L$, the \emph{full trio generated
  by $L$} is the smallest full trio containing $L$, which we denote by
$\T(L)$. Since the intersection of full trios is a full trio
and every full trio includes the regular languages, such a class
indeed exists. Observe that if $L\ne\emptyset$, then $\T(L)$
consists of precisely those languages of the form $RL$, where $R$ is a
rational transduction. In particular, a language $RL$ in $\T(L)$
can be denoted using a representation for $R$. Furthermore, note that
$\T(L)$ is always closed under union, because $RL\cup SL=(R\cup
S)L$ and the union of rational transductions is again a rational
transduction.
\begin{lemma}\label{lem:universality-undecidable}
  The universality problem is undecidable for $\T(L_{ab})$.
\end{lemma}
\begin{proof}
  We reduce Post's Correspondence Problem (PCP) to the (non-)universality problem of
  $\T(L_{ab})$\footnote{It is worth comparing this proof to that of
  \citep[Lemma~8.3]{Berstel:Transductions-context-free-languages:1979:a}}. An
  instance of PCP consists in an alphabet $A$,  with $A\cap\{0,1\}=\emptyset$, together with two morphisms
  $\alpha,\beta\colon A^*\to \{0,1\}^*$. The problem asks whether there is a word
  $w\in A^+$ such that $\alpha(w)=\beta(w)$.

  Suppose we are given a PCP instance. From $\alpha$ and $\beta$, we build a language
  $M_{\alpha,\beta}\in\T(L_{ab})$ over some alphabet $B$ such that $M_{\alpha,\beta}=B^*$ iff there
  exists no $w\in A^+$ with $\alpha(w)=\beta(w)$.
  Our language $M_{\alpha,\beta}$ will be of the form $R_\alpha L_{ab}\cup R_\beta L_{ab}\cup K$, where
  $R_\alpha,R_\beta$ are rational transductions and $K$ is regular. Since
  $\T(L_{ab})$ is effectively closed under union and includes the
  regular languages, it contains $M_{\alpha,\beta}$.

  Given a morphism $\gamma\colon A^*\to\{0,1\}^*$, we first describe the construction
  of $R_\gamma$ (which will be instantiated for $\gamma=\alpha,\beta$ to get $R_\alpha,R_\beta$). We want to construct a rational transduction
  $R_\gamma$ such that
  \begin{equation}
    R_\gamma L_{ab}=\{uv \mid u\in A^+,~v\in\{0,1\}^*,v\ne\gamma(u)
    \}.
    \label{eq:r_gamma-l_ab=uv-mid}
  \end{equation}

  Assume first that we are able to construct such a rational transduction $R_\gamma$ from $\gamma$. We argue that this provides the desired reduction.
  Indeed, let $B=A\cup\{0,1\}$ and let $K=B^* \setminus (A^+\{0,1\}^*)$. As explained
  above, $\T(L_{ab})$ effectively contains $ M_{\alpha,\beta}=R_\alpha L_{ab}~\cup~ R_\beta L_{ab} ~\cup~K$.
  Now, in view of the values of $R_\alpha L_{ab}$ and $R_\beta L_{ab}$ given
  by~\eqref{eq:r_gamma-l_ab=uv-mid},  there is no $w\in A^+$ with $\alpha(w)=\beta(w)$ if and only if $M_{\alpha,\beta}=B^*$.

  It remains to construct a rational transduction $R_\gamma$ satisfying
  \eqref{eq:r_gamma-l_ab=uv-mid}. The intuition is that one uses the input word
  $a^nb^n$ of $L_{ab}$ to produce all words $uv$ where the $n$-th position of $v$ differs
  from the $n$-th position of
  $\gamma(u)$. 
  Observe that $\gamma(u)\neq v$ means that
  \begin{enumerate}
  \item either $|\gamma(u)|> |v|$,
  \item or $|\gamma(u)|< |v|$,
  \item or $\gamma(u)$ and $v$ differ at some common position.
  \end{enumerate}
  Thus, the rational transduction $R_\gamma$ is the union of three rational transductions. We
  illustrate its construction for the third case. As it is well-known
  that rational transductions are realized by finite transducers, \emph{i.e.},
  finite automata reading an input and producing an output along
  transitions~\citep{Berstel:Transductions-context-free-languages:1979:a},
  we describe the third part of $R_\gamma$ as such a transducer. For our
  purpose, it suffices to describe its behavior on inputs of the form
  $a^nb^n\in L_{ab}$. On such an input, the transducer first
  outputs the prefix $u=xyz$ of $uv\in R_\gamma L_{ab}$ in three~steps:
  \begin{itemize}
  \item It outputs some word $x\in A^*$ while, at the same time,
    it reads $a^{|\gamma(x)|}$.
  \item It outputs a symbol $y\in A$ and remembers in its state a
    position $k$ in $\gamma(y)$ and the symbol $i$ in $\gamma(y)$ at position
    $k$.
  \item It outputs some word $z\in A^*$.
  \end{itemize}
  The transducer then produces the suffix $v=rst$ of
  $uv\in R_\gamma L_{ab}$, again in three steps:
  \begin{itemize}
  \item It outputs some word $r\in \{0,1\}^*$ while, at the same time, it reads
    $b^{|r|}$.
  \item It outputs a word $s$ that, at position $k$, has a symbol that
    differs from $i$.
  \item It outputs some word $t\in\{0,1\}^*$.
  \end{itemize}
  Since the input is of the form $a^nb^n$, the above procedure ensures that
  $|\gamma(x)| = |r|$ and therefore produces precisely the words $uv$
  such that $v$ and $\gamma(u)$ differ at some common position.
  Treating the other two cases accordingly yields a construction of
  $R_\gamma$ from $\gamma$, which concludes the proof.
\end{proof}

\begin{proposition}[\citeauthor{Place-personal15}, \citeyear{Place-personal15}]\label{prop:Lab-undecidable}
  Let $\C$ be a full trio that contains the language $L_{ab}$. Then,
  testing whether a given language from $\C$ is piecewise testable is
  undecidable.
\end{proposition}

\begin{proof}
  Since $\T(L_{ab})$ is effectively contained in $\C$, it
  clearly suffices to show that piecewise testability is undecidable
  for $\T(L_{ab})$.

  Let $L$ be a language from $\T(L_{ab})$ over an alphabet $A$.
  We show that, if we can decide whether $L$ is piecewise testable,
  then we can also decide whether $L$ is universal. Since universality
  is undecidable for $\T(L_{ab})$ by Lemma~\ref{lem:universality-undecidable}, we have a contradiction.

  Fix $K$ as any language in $\T(L_{ab})$ that is not piecewise testable and $\#$ as a symbol that is not in $A$. We claim that the
  language $L' = K\#A^* \cup A^*\#L$  is piecewise testable if and
  only if $L$ is universal. Note that $L'$ belongs to
  $\T(L_{ab})$ since the latter is effectively closed under union.

  If $L$ is universal, then $L'=  A^*\#A^*$, which is trivially
  piecewise testable.

  If $L$ is not universal, assume by contradiction that $L'$ is
  piecewise testable. By hypothesis there is a word $w \not\in L$.  It is
  then immediate that $K = L' (\#w)^{-1}$ is also piecewise testable
  (as piecewise testable languages are closed under residuals) which
  is a contradiction by choice of $K$.
\end{proof}

\noindent Since the context-free languages are a full trio and contain
$L_{ab}$, we have the following as one example. However, note that
this is also true for all language classes mentioned in
Section~\ref{sec:feasibility}.
\begin{corollary}\label{cor:cfl-testing}
  Testing if a context-free language is piecewise testable is undecidable.
\end{corollary}

\section{Concluding Remarks}

Since the decidability results we presented are in strong
contrast with the remark of Hunt~III quoted in the introduction, we briefly
comment on this. What we essentially show is that undecidable
emptiness-of-intersection for a class $\C$ does not always imply
undecidability for separability of $\C$ with respect to some
nontrivial class of languages. In the case of separability with
respect to piecewise testable languages, the main reason is basically
that we only need to construct intersections of languages from $\C$
with languages that are regular (or even piecewise testable). Here,
the fact that such intersections can be effectively constructed are,
together with a decidable diagonal problem and some mild
closure properties, sufficient for decidability.

In terms of future work, we see many interesting directions and new questions.
A first set of problems concerns separability by \ptl and very related topics.
Which language classes have a decidable diagonal problem?
Can Theorem~\ref{theo:decidability} be extended to also give complexity guarantees?
In particular, although we now know that separability of context-free languages by
\ptl is decidable, we do not know what the complexity of the problem is. For example, it is not
known if this problem is elementary or not. Finally, can we find similar
characterizations for separability by subclasses of \ptl, as considered in
\citep{HofmanM-icdt15}?

Another interesting set of questions concerns possible generalizations of our result to broader classes of separators.
Recently, it was shown that \ptl is not the only nontrivial class for which separability is decidable beyond regular languages.
In~\cite{DBLP:conf/icalp/ClementeCLP17} and \cite{DBLP:conf/lics/CzerwinskiL17} decidability of separability by \emph{regular languages} was proven
for classes of integer VASS languages and one counter net languages, respectively. Both these language classes are full trios.
Can one formulate natural equivalent conditions (in the style of Theorem~\ref{theo:decidability})
not only for \ptl separability, but for regular separability or separability by other classes?
For example, all the cases of which we are aware seem to confirm the following conjecture: under mild robustness conditions, \emph{full trios have decidable regular separability if and only if they have decidable intersection emptiness problem}. Such a result would again be close to the remark of Hunt~III that we quoted in the introduction and would make it more precise.



\segment{Acknowledgments} We would like to thank Tom\'a\v{s}
Masopust for pointing us to \citep{DBLP:journals/jacm/Hunt82a} and
Thomas Place for pointing out to us that determining if a given
context-free language is piecewise testable is undecidable. We are
also grateful to the anonymous reviewers for many helpful remarks that
simplified proofs. 

\bibliographystyle{abbrvnat}
\bibliography{citat_new}

\end{document}